\documentclass[letterpaper,onecolumn,accepted=2024-11-09]{quantumarticle}
\pdfoutput=1
\usepackage[utf8]{inputenc}
\usepackage{fullpage}
\usepackage{braket}

\usepackage{amsmath,amsfonts,amsthm,bm}

\def\eqref#1{equation~\ref{#1}}

\def\1{\bm{1}}

\DeclareMathAlphabet{\mathsfit}{\encodingdefault}{\sfdefault}{m}{sl}
\SetMathAlphabet{\mathsfit}{bold}{\encodingdefault}{\sfdefault}{bx}{n}

\newcommand{\E}{\mathbb{E}}

\newcommand{\R}{\mathbb{R}}

\DeclareMathOperator*{\argmax}{arg\,max}

\usepackage{subfigure}

\usepackage{latexsym}
\usepackage{graphics}
\usepackage{amsmath}
\usepackage{xspace}
\usepackage{amssymb}
\usepackage{psfrag}
\usepackage{epsfig}
\usepackage{amsthm}
\usepackage{pst-all}

\usepackage{color}

\definecolor{Red}{rgb}{1,0,0}
\definecolor{Blue}{rgb}{0,0,1}
\definecolor{Olive}{rgb}{0.41,0.55,0.13}
\definecolor{Yarok}{rgb}{0,0.5,0}
\definecolor{Green}{rgb}{0,1,0}
\definecolor{MGreen}{rgb}{0,0.8,0}
\definecolor{DGreen}{rgb}{0,0.55,0}
\definecolor{Yellow}{rgb}{1,1,0}
\definecolor{Cyan}{rgb}{0,1,1}
\definecolor{Magenta}{rgb}{1,0,1}
\definecolor{Orange}{rgb}{1,.5,0}
\definecolor{Violet}{rgb}{.5,0,.5}
\definecolor{Purple}{rgb}{.75,0,.25}
\definecolor{Brown}{rgb}{.75,.5,.25}
\definecolor{Grey}{rgb}{.5,.5,.5}

\newcommand{\bfPhi}{{\bf\Phi}}

\newcommand{\pr}{\mathbb{P}}

\newcommand{\C}{\mathbb{C}}
\newcommand{\G}{\mathbb{G}}

\newcommand{\rmCAT}{{\rm CAT}}

\newcommand{\ignore}[1]{\relax}

\newtheorem{theorem}{Theorem}[section]

\newtheorem{lemma}[theorem]{Lemma}

\newtheorem{prop}[theorem]{Proposition}

\newtheorem{coro}[theorem]{Corollary}

\newtheorem{Defi}[theorem]{Definition}

\newcommand{\ER}{Erd{\"o}s-R\'{e}nyi }

\definecolor{Red}{rgb}{1,0,0}
\definecolor{Blue}{rgb}{0,0,1}
\definecolor{Olive}{rgb}{0.41,0.55,0.13}
\definecolor{Green}{rgb}{0,1,0}
\definecolor{MGreen}{rgb}{0,0.8,0}
\definecolor{DGreen}{rgb}{0,0.55,0}
\definecolor{Yellow}{rgb}{1,1,0}
\definecolor{Cyan}{rgb}{0,1,1}
\definecolor{Magenta}{rgb}{1,0,1}
\definecolor{Orange}{rgb}{1,.5,0}
\definecolor{Violet}{rgb}{.5,0,.5}
\definecolor{Purple}{rgb}{.75,0,.25}
\definecolor{Brown}{rgb}{.75,.5,.25}
\definecolor{Grey}{rgb}{.5,.5,.5}
\definecolor{Pink}{rgb}{1,0,1}
\definecolor{DBrown}{rgb}{.5,.34,.16}
\definecolor{Black}{rgb}{0,0,0}

\usepackage{float}

\usepackage{amssymb}
\usepackage{mathtools}	
\usepackage{amsthm}
\usepackage{amsmath}

\usepackage{hyperref}
\usepackage{cleveref}

\usepackage{graphicx}
\usepackage{xcolor}

\usepackage{tikz}
\usetikzlibrary{shapes, positioning, calc}

\theoremstyle{plain}

\newcommand{\change}[1]{{#1}}

\title{Combinatorial NLTS From the Overlap Gap Property}

\author{Eric R.\ Anschuetz}
\email{eans@caltech.edu}
\affiliation{MIT Center for Theoretical Physics, Cambridge, MA 02139, USA}
\author{David Gamarnik}
\email{gamarnik@mit.edu}
\affiliation{Sloan School of Management, MIT, Cambridge, MA 02139, USA}
\author{Bobak Kiani}
\email{bkiani@g.harvard.edu}
\affiliation{Department of Electrical Engineering and Computer Science, MIT, Cambridge, MA 02139, USA}

\begin{document}

\maketitle

\begin{abstract}
In an important recent development, Anshu, Breuckmann, and Nirkhe~\cite{anshu2022nlts} resolved positively the so-called No Low-Energy Trivial State (NLTS) conjecture by Freedman and Hastings~\cite{freedman2014quantum}. The conjecture postulated the existence of linear-size local Hamiltonians on $n$ qubit systems for which no near-ground state can be prepared by a shallow (sublogarithmic depth) circuit. The construction in \cite{anshu2022nlts} is based on recently developed good quantum codes. Earlier results in this direction included the constructions of the so-called Combinatorial NLTS---a weaker version of NLTS---where a state is defined to have low energy if it violates at most a vanishing fraction of the Hamiltonian terms~\cite{eldar2017local},\cite{anshu2022construction}. These constructions were also based on codes. 

In this paper we provide a ``non-code'' construction of a class of Hamiltonians satisfying the Combinatorial NLTS. The construction is inspired by one in~\cite{anshu2022construction}, but our proof uses the complex solution space geometry of random K-SAT instead of properties of codes. Specifically, it is known that above a certain clause-to-variables density the set of satisfying assignments of random K-SAT exhibits the overlap gap property, which implies that it can be partitioned into exponentially many clusters each constituting at most an exponentially small fraction of the total set of satisfying solutions~\cite{AchlioptasCojaOghlanRicciTersenghi}. We establish a certain robust version of this clustering property for the space of near-satisfying assignments and show that for our constructed Hamiltonians every combinatorial near-ground state induces a near-uniform distribution supported by this set. Standard arguments then are used to show that such distributions cannot be prepared by quantum circuits with depth $o(\log n)$. 

Unfortunately, our Hamiltonian is only local on average due to fluctuations in the numbers of clauses in which variables are involved. On the positive side, since the clustering property is exhibited by many random structures---including proper coloring, maximum cut, and the hardcore model---we anticipate that our approach is extendable to these models as well, including certain classical models of bounded degree. We sketch such an extension here for the $p$-spin Ising model on a sparse random regular graph as a non-code example where exactly local Combinatorial NLTS may be achieved. 
\end{abstract}

\section{Introduction}
A remarkable feature exhibited by some quantum systems, and not found in  classical models, is the potential complexity of the  state representation. 
Any classical state on an $n$-bit system is just an $n$-bit string and thus  requires at most $n$ bits to be represented. Quantum systems on $n$ qubits, 
on the other hand, are described in general by $2^n$ associated complex-valued 
amplitudes and thus in principle lead to a representation complexity which is exponential
in $n$. This should not necessarily be the case, however, for ``physical'' quantum states such as those arising as ground states of succinctly described
Hamiltonians. The so-called No Low-Energy Trivial States conjecture by Freedman and Hastings~\cite{freedman2014quantum} posits exactly this question:
does there exist a simple-to-describe Hamiltonian for which any associated near-ground state is too complex to describe? As a proxy for descriptional 
simplicity they propose to focus on Hamiltonians which are sums of linearly many local Hamiltonians.
Specifically, these are Hamiltonians $H$ acting on $n$ qubits, which can be written
as $H=\sum_i H_i$, where the number of terms $H_i$ is linear in $n$,  each $H_i$ has  bounded by $O(1)$ operator norm, and 
is local, in the sense that each $H_i$ acts on $O(1)$-many qubits.
As a proxy for the complexity measure of  state representation, we say that a state $\rho$ is  complex if it cannot be prepared by a sublogarithmic-in-$n$ depth
quantum circuit. Namely, $\rho$ is complex if there does not exist a logarithmic-depth circuit such that the  unitary $U$ associated with this circuit satisfies 
$\rho = U |0\rangle \langle 0 |U^\dagger$. The NLTS conjecture then posits the existence of linear-size local Hamiltonians $H$ such that every state $\rho$ which is a near-ground state of $H$ (that is, a near minimizer of $\min {\rm Tr}(H\rho)$) is complex in the sense above.

The existence of such Hamiltonians is a necessary consequence of the validity of the quantum PCP conjecture: if a near-ground state were of low complexity, a classical representation of it could be obtained in polynomial time classically. This would then imply that the complexity class $\textsc{QMA}$ satisfies $\textsc{QMA}=\textsc{NP}$, 
which is widely
believed not to be the case. Thus, if the quantum PCP conjecture is valid, NLTS must hold. 
Validating the NLTS conjecture was then proposed as a reasonable step towards the ultimate goal of proving the 
quantum PCP theorem~\cite{aharonov2013guest}.

A series of results~\cite{eldar2017local,eldar2019robust,bravyi2020obstacles,anshu2020circuit} gave strong evidence in favor of the conjecture, with the strongest evidence (before its eventual resolution) given by Anshu and Breuckmann~\cite{anshu2022construction}. There, the validity of the so-called \emph{combinatorial} version of NLTS---the subject of this paper---was established. Given a linear-size local Hamiltonian $H=\sum_i H_i$ and given $\epsilon>0$,  a state $\rho$ is an $\epsilon$-near ground state of $H$ in the combinatorial sense if it is a ground state for at least a $(1-\epsilon)$ fraction of Hamiltonian terms $H_i$. Namely, using constraint satisfaction terminology, $\rho$ has to satisfy
all but an $\epsilon$-fraction of constraints $H_i$. Validity of the NLTS conjecture trivially implies its combinatorial version but not vice-versa.

A complete positive resolution of the original NLTS conjecture was finally obtained in a remarkable paper by  
Anshu, Breuckmann, and  Nirkhe~\cite{anshu2022nlts}. Their construction as well
as all of the previous constructions (to the best of our knowledge) are based on codes. Specifically, classical or quantum codes are used to design the Hamiltonians
$H_i$. Near-code words are then shown to exhibit sufficiently strong entanglement which obstructs their low complexity representation. The construction
in \cite{anshu2022nlts} leverages a recently constructed quantum Tanner code~\cite{leverrier2022quantum} to construct a Hamiltonian which exhibits the following
clustering property: near-ground states induce a probability measure on a basis which is supported on a disjoint union of exponentially many clusters 
separated by a linear distance. Most of the time this was shown directly for the computational basis, as in~\cite{anshu2022construction} and~\cite{eldar2017local}.
On the other hand, the construction in~\cite{anshu2022nlts} is shown to exhibit this clustering in the computational and/or Hadamard bases.

Remarkably, such a clustering property is ubiquitous in a completely  
different research direction devoted to studying random constraint satisfaction problems and closely related models known as 
diluted spin glasses. Examples of the former include random instances of the K-SAT problem~\cite{AchlioptasCojaOghlanRicciTersenghi}, 
the problem of finding a proper coloring of a graph~\cite{achlioptas2008algorithmic},\cite{mezard2005clustering}, the largest
independent set problem~\cite{coja2011independent}, and the maximum cut
problem on sparse random hypergraphs~\cite{chen2019suboptimality}. An example of the latter includes an Ising or Edwards--Anderson model supported on a sparse random graph~\cite{FranzLeone}.
Many of these models are known to exhibit an \emph{overlap gap property} (OGP), which roughly speaking states that any two (classical) near-ground states
are either close or far from each other in Hamming distance. The overlap gap property was then used to prove the clustering 
property: the set of near ground states is supported on a disjoint union of exponentially many clusters each consisting of at most an exponentially
small fraction of the total set of states. Starting from \cite{gamarnik2014limits}
this property was then widely used to establish various algorithmic complexity lower bounds, 
in particular
ruling out various classical algorithms for finding near-ground states. The links between the overlap gap property and algorithmic complexity 
are surveyed in \cite{gamarnik2021overlap} and \cite{gamarnik2022disordered}. The abundance of clustering properties exhibited in these models
raises the natural question of whether 
it is possible to use these models for constructing NLTS. Our paper provides a positive answer to this question for the case of Combinatorial NLTS. 
Specifically, we use the appropriately extended clustering property of the random K-SAT model toward this goal.
Naturally, the diagonal Hamiltonians arising directly from the K-SAT model cannot be used for NLTS directly 
since low-energy states of such Hamiltonains trivially include those with small description complexity. Instead,
we resort to a method similar to one used in \cite{anshu2022construction}. Specifically, using random instances of the K-SAT model we construct a quantum Hamiltonian
such that  each combinatorial near-ground state of it induces a probability measure on the computational basis which is
a  near-uniform distribution on the space of near-satisfying assignments of the K-SAT formula. This is obtained by projecting a tensor product
of cat-like states onto the set of nearly satisfying assignments of the formula.
We then use a standard argument to show that such measures cannot be prepared by shallow quantum circuits. Toward this goal we utilize a certain
appropriate mixture of computational and Hadamard bases on the one hand, and prove a certain ``robust'' version of the overlap gap and clustering property
on the other hand. These constitute the main technical contributions of our paper. 

Unfortunately, the Hamiltonians we construct are local only on average, in the sense that typical Hamiltonian terms act on at most $O(1)$ many qubits, 
while some may act on as many as $O(\log n)$ qubits.
This a direct implication of the fact that while the average degree of a sparse random graph has constant degree, the largest degree grows like $O(\log n)$
w.h.p.\
We conjecture, however, that a similar  version of the OGP also holds for certain classical models of bounded degree, with a similar quantum extension to Combinatorial NLTS as we exhibit here. This appears to require substantially more involved technical work which we do not pursue, though we give an outline of a proof in Section~\ref{sec:bounded_degree}.

On a positive note, our work shows that code
structures such as those that underlie all of the prior NLTS constructions are not essential for the goal of building NLTS Hamiltonians, 
and in fact quantum Hamiltonians exhibiting Combinatorial NLTS can be constructed 
from classical random Hamiltonians supported by random graphs
and hypergraphs. 
Indeed we conjecture that constructions of Combinatorial NLTS similar to ours can be obtained from other models exhibiting the overlap gap property,
including the largest independent set problem as in~\cite{gamarnik2014limits} and the maximum cut problem as in~\cite{chen2019suboptimality}. 
We further conjecture that the class of Hamiltonians we construct based on random K-SAT instances is in fact NLTS (and not just Combinatorial NLTS), but we are not able
to prove this. Finally, the potential ramification of our result for validating the quantum PCP conjecture is another interesting question for further research.

The remainder of the paper is structured as follows. Random K-SAT instances and the statement of the Combinatorial NLTS are introduced
in the next section. Our main result is stated in the next section as well. In  Section~\ref{section:OGP+clustering} we  define the overlap gap property and the implied 
clustering properties. In the same section we  state and prove  a robust version of the overlap gap and clustering property which applies to 
``nearly'' satisfying assignments of ``nearly complete'' subsets of Boolean variables. This is an  extension of an older result by Achlioptas, Coja-Oghlan 
and Ricci-Tersenghi~\cite{AchlioptasCojaOghlanRicciTersenghi} which regards assignments with no clause violations. 
There both the overlap gap  and the clustering properties are established for $K\ge 8$. The robustness required for our purposes forces us
to resort instead to large $K$ and the details of the bounds are somewhat delicate.
We complete the proof of our main result in Section~\ref{section:main-result}. Finally, we end in Section~\ref{sec:bounded_degree} with a sketch of an extension of our results to models with bounded degree.

\section{Model, construction, and the main result}

\subsection{Random K-SAT formulas}
We begin by describing a random instance of the Boolean constraint satisfaction problem, specifically the K-SAT problem. 
\change{A Hamiltonian is constructed from this classical constraint satisfaction problem which is different but nonetheless related to instances of problems studied in so-called quantum SAT \cite{bravyi2006efficient}.} 
An integer $K$ is fixed.
A collection of Boolean variables $x_1,\ldots,x_n$ is fixed, each taking values in $\{0,1\}$. A $K$-clause $C$ over these variables is a disjunction
of the form $y_1\vee y_2\vee\cdots \vee y_K$, consisting of exactly $K$ elements where each $y_j, j\in [k]$ is either one of $x_1,\ldots,x_n$
or one of the negations $\neg x_1,\ldots,\neg x_n$. For example, a $K=4$-clause with $y_1=x_3,y_2=\neg x_{17}, y_3=\neg x_6, y_4=\neg x_2$,
gives rise to the clause $\left(x_3\vee \neg x_{17} \vee \neg x_6 \vee \neg x_2\right)$.

A formula $\Phi$ is a conjunction $\Phi\triangleq C_1\wedge \cdots  \wedge C_m$ of any collection of $K$-clauses $C_1,\ldots,C_m$.
A solution is any assignment $x=(x_1,\ldots,x_n)\in \{0,1\}^n$. A solution $x$ is said to satisfy a clause $C$ if the Boolean evaluation of this clause is $1$, namely, if at least one of the variables $x_i$ appearing in this clause satisfies it. For the example $4$-clause above, any solution with $x_3=1$ satisfies this clause, and so does any solution with $x_{17}=0,x_6=0,x_2=0$. Notice that for each clause $C$ there exists a unique assignment of the $K$ variables in this clause which does not satisfy it which we denote by $v(C)\in \{0,1\}^K$. A solution $x$ satisfies the formula $\Phi$ if it satisfies every clause in the formula.

For every clause $C_j$ we denote by $(j,1),\ldots,(j,K)$ the indices of variables appearing in $C_j$ for convenience. For example, if, say,
$C_5=\left(\neg x_4\vee \neg x_{1} \vee \neg x_9 \vee  x_{12}\right)$ and $C_8=\left(x_{32}\vee  x_{17} \vee \neg x_{16} \vee \neg x_{12}\right)$, then
$(5,1)=4,(5,2)=1,(5,3)=9,(5,4)=12$ and $(8,1)=32,(8,2)=17,(8,3)=16,(8,4)=12$.

We now turn to the description of a (uniform) random instance of a K-SAT problem. For a fixed $K,n,m$ we consider the space of all $(2n)^K$ clauses 
(with repetition of variables within a clause allowed and the literals within a clause ordered). The random formula is generated by selecting $C_1,\ldots,C_m$
independently and uniformly at random from this space of all clauses. An instance of a random formula is denoted by $\bfPhi(n,m)$, where the dependence on $K$
is kept implicit. We consider a proportional setting where a ratio $m/n$ is fixed to be a constant as $n$ increases. 
A relevant regime for us will be when $m/n$ is of the form $\alpha 2^K \log 2$ for some $\alpha\in (1/2,1)$. In this case the set of satisfying assignments
exhibits a clustering property crucial for our construction. The derivation of this clustering property is the subject of Section~\ref{section:OGP+clustering}.

Next we turn to a construction of a sparse Hamiltonian based on an instance of a K-SAT model.

\subsection{Construction of a quantum Hamiltonian and the main result}
\begin{table}[]
\centering
\change{
\begin{tabular}{ll}
\hline
Symbol & Description \\ \hline
     $x_i$  &  Boolean variable           \\
     $n$  &  Number of SAT variables           \\
    $m$   &   Number of clauses          \\
   $\Phi(n,m)$    &      Random SAT formula with $n$ variables and $m$ clauses       \\
  $C_i$     &  Clause $i$ in random formula $\Phi(n,m)$           \\
  $C(x_i)$ & Set of clauses containing variable $x_i$ \\
   $K$    &   Locality of clauses          \\
    $N$   &   Number of qubits ($N=Km$)      \\
   $v(C_j)$    &   Unique non-satisfying
solution of clause $C_j$         \\ 
$D(i)$ & Subset of qubits corresponding to the variable $x_i$ \\
\hline
\end{tabular}
\caption{Description of notation used in construction of the quantum Hamiltonian (\Cref{eq:Hamiltonian}).}
}
\end{table}

We now turn to the description of a quantum Hamiltonian which will give a rise to a combinatorial NLTS construction. This construction is inspired by the one in \cite{anshu2022construction}.
We fix a realization of  a random formula $\bfPhi(n,m)$ on variables $x_1,\ldots,x_n$ and clauses
$C_1,\ldots,C_m$ which we also denote by $\bfPhi(n,m)$ for convenience. For each $j=1,\ldots,m $ the associated variables are denoted by 
$x_{j,k}, k\in [K]$. That is, each $x_{j,k}$ is one of $x_1,\ldots,x_n$.
We consider a quantum system on $m K$ qubits. These qubits are denoted by 
 $(1,1),\ldots,(1,K), (2,1), \ldots,(2,K), \ldots,(m ,1), \ldots, (m ,K)$, where we recall that $(j,k)$ denotes the index of the $k$-th variable
 appearing in clause $C_j$. For each  $i\in [n]$ we denote by $D(i)$ the subset of qubits corresponding to the variable $x_i$.
 Namely, it is the set of all qubits with indices $(j,k)$ such that $(j,k)=i$. As an example, if the variable $x_3$ appears in the clause
$C_4$ in position $2$, $C_7$ in position $5$, and $C_{12}$ in position $1$---and does not appear in any other clause---then
$D(x_3)=\{(4,2),(7,5),(12,1)\}$. Note that the sets $D(i)$ are disjoint and that their union is the set of all qubits $\{(j,k), j\in [m] , k\in [K]\}$. For each $i \in [n]$, we also denote the set of clauses containing variable $x_i$ by $C(x_i) = \{ C_j: i \in \{ (j,1),\ldots,(j,K) \} \}$. For example, if there are three clauses: $C_1 = \left(\neg x_4\vee \neg x_{1} \vee \neg x_9 \vee  x_{7}\right)$, $C_2 = \left(\neg x_5\vee  x_{2} \vee \neg x_1 \vee  x_{3}\right)$, and $C_3 = \left(\neg x_3\vee \neg x_{8} \vee  x_4 \vee  x_{6}\right)$, then $C(x_1) = \{ C_1, C_2 \}$.

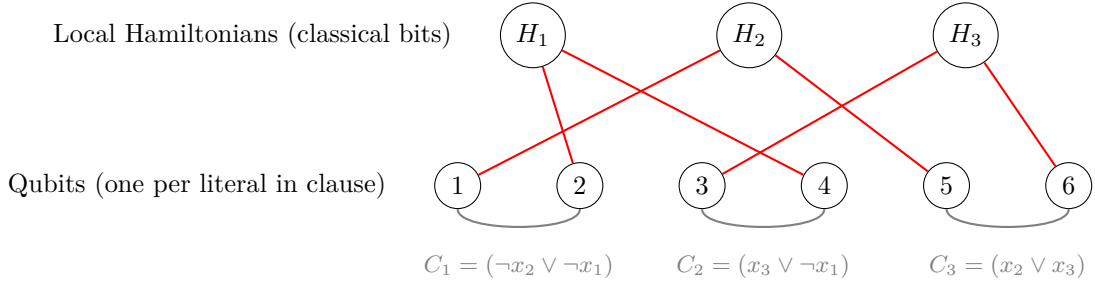
\begin{figure}
    \centering
    \begin{tikzpicture}[
    every node/.style={draw, circle},
    red edge/.style={draw, red, thick},
    gray edge/.style={draw, gray, thick},
    edge label/.style={font=\small, midway, auto, outer sep=-25pt, draw=none},
]

\node (h1) at (1, 2) {$H_1$};
\node[right=2cm of h1] (h2) {$H_2$};
\node[right=2cm of h2] (h3) {$H_3$};

\node (q1) at (0, 0) {1};
\node[right=1cm of q1] (q2) {2};
\node[right=1cm of q2] (q3) {3};
\node[right=1cm of q3] (q4) {4};
\node[right=1cm of q4] (q5) {5};
\node[right=1cm of q5] (q6) {6};

\node[left=0.5cm of h1, draw=none] {Local Hamiltonians (classical bits)};
\node[left=0.5cm of q1, draw=none] {Qubits (one per literal in clause)};

\draw[red edge] (h1) -- (q2);
\draw[red edge] (h1) -- (q4);
\draw[red edge] (h2) -- (q1);
\draw[red edge] (h2) -- (q5);
\draw[red edge] (h3) -- (q3);
\draw[red edge] (h3) -- (q6);

\draw[gray edge, out=270, in=270, looseness=0.5] (q1) to node[edge label, below] {$C_1 = \left(\neg x_2\vee \neg x_{1}\right)$} (q2);
\draw[gray edge, out=270, in=270, looseness=0.5] (q3) to node[edge label, yshift=-3pt, below] {$C_2 = \left(x_3\vee \neg x_{1}\right)$} (q4);
\draw[gray edge, out=270, in=270, looseness=0.5] (q5) to node[edge label, yshift=-6pt, below] {$C_3 = \left( x_2\vee x_{3}\right)$} (q6);

\end{tikzpicture}
    \vspace{-1.3cm}
    \caption{A visualization of an example Hamiltonian $H = H_1 + H_2 + H_3$ from (\ref{eq:Hamiltonian}) constructed from a classical 2-SAT formula of 3 clauses where $C_1 = \left(\neg x_2\vee \neg x_{1}\right)$, $C_2 = \left(x_3\vee \neg x_{1}\right)$, and $C_3 = \left( x_2\vee x_{3}\right)$. Here, each term of the Hamiltonian $H_i$ as described in (\ref{eq:Hamiltonian}) is $4$-local as the $\ket{\rmCAT(i)}$ states are each supported on two qubits on disjoint clauses which correspond to the matrices $Q_{C_j,\gamma}^{-1}$ that are $2$-local.}
    \label{fig:hamiltonian_viz}
\end{figure}

For each $i\in [n]$ we denote by $\ket{\rmCAT(i)}$ the CAT state associated with qubits in $D(i)$. That is, 
$\ket{\rmCAT(i)}=\frac{1}{\sqrt{2}} |00\ldots 0\rangle +\frac{1}{\sqrt{2}} |11\ldots 1\rangle $ with qubits supported on $D(i)$.
We denote by $\ket{\rmCAT}$ the tensor product of these states: $\ket{\rmCAT}\triangleq   \bigotimes_{i \in [n]}\ket{\rmCAT(i)}$.

Next, for each $j\in [m ]$ we consider the projection onto the space spanned by vectors orthogonal to $|v(C_j)\rangle$, 
where, as we recall, $v(C_j)$ is the unique non-satisfying
solution of the clause $C_j$. That is, we consider the operator:
\begin{align*}
Q_{C_j}\change{\triangleq}\sum_{v\in \{0,1\}^K, v\ne v(C_j)} \ket{v} \bra{v}
\end{align*}
applied to the subset $\{(j,1),\ldots,(j,K)\}$ of $\change{N=Km}$ qubits corresponding to the clause $C_j$, tensored with the identity operator on the remaining qubits. 
\change{For} any element $|z\rangle$ of the computational basis we have that $Q_{C_j} |z\rangle =|z\rangle$ if $z\in \{0,1\}^{\change{N}}$ encodes a satisfying solution of $C_j$ and $Q_{C_j} |z\rangle =0$ otherwise. \change{Thus, indices of $\ket{z}$ can be interpreted as assigning values to the classical variables $x_1, \dots, x_n$. Note that most $\ket{z}$ will not correspond to a globally valid assignment to $x_1, \dots, x_n$, as any valid assignment must be consistent with the projectors onto the CAT states; states corresponding to such invalid assignments will have eigenvalue $0$ due to these projectors.}

The operator $Q_{C_j}$ is not invertible due to the existence of non-satisfying solutions $z$. As in~\cite{anshu2022construction} we consider an approximating
but invertible operator by fixing a constant $\gamma>0$ and defining:
\begin{align*}
Q_{C_j,\gamma}\triangleq (1-\gamma)Q_{C_j}+\gamma I,
\end{align*}
acting on the same set of $K$ qubits again tensored with the identity operator on the remaining qubits. 
We also define an invertible operator
\begin{align*}
\change{
Q(\gamma)\triangleq \prod_{1\le j\le m}Q_{C_j,\gamma}.
}
\end{align*}
Next, for each variable $x_i, i\in [n]$
we consider the tensor product:
\begin{align*}
\change{
Q_{x_i,\gamma}\triangleq \prod_{j: C_j\in C(x_i)} Q_{C_j,\gamma},
}
\end{align*}
again tensored with the identity on the remaining qubits. Namely, we take a tensor product of operators $Q_{C_j,\gamma}$ associated with clauses $C_j$
containing variable $x_i$. 
This is an invertible operator on the set 
of all $\change{N}$ qubits; $Q_{x_i,\gamma}^{-1}$ denotes its inverse.

Finally, we consider the Hamiltonian:
\begin{align}
H_i & = Q_{x_i,\gamma}^{-1}\left (I-\ket{\rmCAT(i)} \bra{\rmCAT(i)}\right) Q_{x_i,\gamma}^{-1}, \qquad i\in [n]  \notag \\
H & = \sum_{i\in [n]} H_i. \label{eq:Hamiltonian}
\end{align}
\change{Intuitively, each $H_i$ penalizes the energy of a state $\ket{\phi}$ if it violates clauses in $C(x_i)$ for the $i$-th variable or has support over computational bases $\ket{z}$ which are not consistent with the CAT projectors.}
We note that $\ket{\rmCAT(i)} \bra{\rmCAT(i)}$ is the projection operator onto the state $\ket{\rmCAT(i)}$.
The operator $H$ is ``local on average'' when $m=O(n)$  as each summand $H_i$ acts on \change{approximately $K^2m/n$} qubits in expectation. Furthermore, the largest locality---namely, the largest number of qubits per $H_i$---is at most $O(\log n)$ in this case, as it corresponds to a maximum of $O(n)$ binomial random variables with $O(n)$ trials and success probabilities $O(n^{-1})$.

This completes the construction of our local (on average) Hamiltonian $H$. We claim that  a suitable choice of $\alpha$, $K$, and $\gamma$
verifies the requirements of the combinatorial NLTS. Specifically, we fix $\epsilon$ and define a combinatorial notion of a state being near the ground state. 

\begin{Defi}\label{definition:Combinatorial-near-ground-state}
A pure state $|\psi\rangle$ on $\change{N=Km}$ qubits is an $\epsilon$-near ground state (in the combinatorial sense) if there exists \change{any} subset $S\subset [n]$\change{---termed an \emph{associated subset}---}such that $|S|\ge (1-\epsilon)n$ and 
such that $|\psi\rangle$ is a ground state for all $H_i$ with $i\in S$. Namely, $H_i | \psi\rangle=0$ for all $i\in S$.
\end{Defi}

\change{We note that a near ground state $\ket{\psi}$ can have multiple associated subsets $S \subset [n]$ that meet the criteria in the above definition.}

Next, we fix a positive integer $T$ and consider
depth-$T$ circuits. Such circuits are defined inductively as follows. A depth-$1$ circuit is a unitary $U_1$ obtained as a tensor product of unitary
operators each acting on a disjoint union of $2$-qubit subsets. For any $t\le T-1$, suppose $U_t$ is any depth-$(t-1)$ circuit. Then, a depth-$t$ circuit is defined
as any  circuit of the form $V U_{t-1}$, where $U_{t-1}$ is any depth-$(t-1)$ circuit and $V$ is any depth-$1$ circuit. 

\begin{Defi}\label{definition:T-trivial}
Given $T>0$, a state $|\psi\rangle$ is $T$-trivial if there exists a depth-$T$ circuit $U_T$ such that $|\psi\rangle = U_T |0\rangle^{\otimes \change{N}}$. 
\end{Defi}

We are now ready to state our main result.
\begin{theorem}\label{theorem:MainResult}
There exist $\alpha,K,\gamma,\epsilon,\tau>0$ such that the following holds w.h.p.\ as $n\to\infty$: the Hamiltonian $H$ is frustration free (i.e., it has ground state energy $\lambda_{\min}(H)=0$), and for every sufficiently large $n$, if $|\psi\rangle$
is an $\epsilon$-near ground state with respect to the Hamiltonian (\ref{eq:Hamiltonian}), then for $T\le \tau \log n$ it holds that $|\psi\rangle$ is not $T$-trivial.
\end{theorem}

\section{Robust overlap gap property and clustering in random K-SAT instances}\label{section:OGP+clustering}

In this section we derive certain structural properties of the solution space geometry of near-satisfying solutions of the random K-SAT problem. 
In particular, we will show that it exhibits the so-called overlap gap property (OGP); namely, the clustering of near-satisfying assignments
when the parameters $\alpha$ and $K$ are set appropriately. This clustering property will be pivotal in the proof of our main result, Theorem~\ref{theorem:MainResult}.

\begin{Defi}
Given $0<\nu_1<\nu_2<1/2$ and a subset $A\subset \{0,1\}^n$ we say that 
$A$ exhibits an $(\nu_1,\nu_2)$-overlap-gap-property---abbreviated as a $(\nu_1,\nu_2)$-OGP---if for every $x^1,x^2\in A$ either $n^{-1}d(x^1,x^2)\le \nu_1$ or $n^{-1}d(x^1,x^2)\ge \nu_2$, where $d(\cdot,\cdot)$ denotes
the Hamming distance in $\{0,1\}^n$. 
\end{Defi}

\begin{figure}
    \centering
    \includegraphics{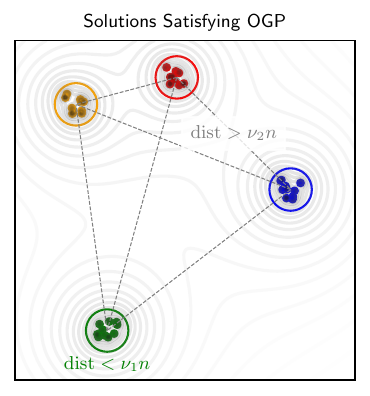}
    \caption{An illustrative visualization of the $(\nu_1,\nu_2)$-OGP shown here as a contour plot for an idealized optimization landscape. Solutions are clustered in balls of radius at most $\nu_1n$ and each pair of disjoint clusters is at least distance $\nu_2 n$ apart.}
    \label{fig:ogp_viz}
\end{figure}

When a set $A$ exhibits an $(\nu_1,\nu_2)$-OGP and additionally $\nu_1<(1/2)\nu_2$,
there exists a unique partitioning of $A$ into a disjoint union of $L$ clusters \change{
$A=\bigcup_{1\le \ell\le L} CL_\ell$}, $CL_{\ell_1}\cap CL_{\ell_2}=\emptyset, \ell_1\ne \ell_2$ 
such that for every $\ell$ and every $x^1,x^2\in CL_\ell$, it holds that
$n^{-1}d(x^1,x^2)\le \nu_1$ and for every $\ell_1\ne \ell_2$ and every $x^j\in CL_{\ell_j}, j=1,2$, $n^{-1}d(x^1,x^2)\ge \nu_2$. 
According to this definition the possibility of a unique
cluster---that is, $L=1$---is allowed.
This clustering property is obtained simply as follows. Define equivalency classes by saying $x\sim y$ iff $n^{-1}d(x,y)\le \nu_1$.
The reflexivity and symmetry are immediate. The transitivity is obtained as well: if $n^{-1}d(x,y), n^{-1}d(y,z) \le \nu_1$ then
$n^{-1}d(x,z)\le 2\nu_1<\nu_2$ and thus, by $(\nu_1,\nu_2)$-OGP, $n^{-1}d(x,z)\le \nu_1$. We call this the $(\nu_1,\nu_2)$-clustering of $A$.

Given a K-SAT formula $\Phi$ on variables with indices $[n]$ and clauses $C_1,\ldots,C_m$, we denote by ${\rm SAT}(\Phi)\subset \{0,1\}^n$ the (possibly empty) set of 
satisfying solutions of $\Phi$. Furthermore, for every $r\ge 0$ let ${\rm SAT}(\Phi,r)$ denote the set of solutions in $\{0,1\}^n$ which
violate at most $r$ clauses so that ${\rm SAT}(\Phi,0)={\rm SAT}(\Phi)$. 
Given a formula $\Phi$ and a subset $S\subset [n]$, let $C(S)$ be the set of clauses in $\Phi$ such that each clause in $C(S)$ consists entirely of variables
$x_i$ with  $i\in S$. This should not be confused with $C(x_i)$, which denotes the set of clauses containing the variable $x_i$.
Similarly, for every $r\ge 0$ let ${\rm SAT}(\Phi(C(S),r))\subset \{0,1\}^n$ denote the set of solutions in $\{0,1\}^n$ which violate
at most $r$ clauses among those in $C(S)$.
Given $\epsilon>0$ we denote by ${\rm SAT}(\Phi,\epsilon,r)$ the set $\bigcup_{S:|S|\ge (1-\epsilon)n}{\rm SAT}(\Phi(C(S),r))$
and let $Z(\Phi,\epsilon,r)=|{\rm SAT}(\Phi,\epsilon,r)|$.

Our main technical result is as follows. 

\begin{theorem}\label{theorem:clustering}
Let $\alpha\in (1/2+1/5,1)$. For large enough $n$, there exist $0<\nu_1<(1/2)\nu_2<\nu_2<1/2$ and $c_2>c_1>0$ such that the following holds.
There exist large enough $K$ and small enough $\bar \epsilon,\bar \lambda $ such that for every $0<\epsilon \le \bar \epsilon, 0\le \lambda\le \bar\lambda$
 the set ${\rm SAT}(\bfPhi(n,\alpha 2^K(\log 2)n),\epsilon,\lambda n)$ satisfies $(\nu_1,\nu_2)$-OGP w.h.p.\
as $n\to\infty$. Furthermore,  also w.h.p. as $n\to\infty$,
\begin{align*}
| {\rm SAT}(\bfPhi(n,\alpha 2^K(\log 2)n))|=| {\rm SAT}(\bfPhi(n,\alpha 2^K(\log 2)n),0,0)| \ge \exp( c_2 n)
\end{align*}
and the associated unique $(\nu_1,\nu_2)$-clustering $CL_\ell$ of 
${\rm SAT}(\bfPhi(n),\epsilon ,\lambda n)$ satisfies 
\begin{align*}
|CL_\ell |\le \exp(c_1 n),
\end{align*}
for all $\ell$.
\end{theorem}
The theorem roughly speaking says that the set of all nearly optimal assignments of every nearly complete  subset of  variables
(with respect to the  induced set of clauses) exhibits an OGP. Furthermore, the implied clustering of this set 
has the property that each cluster constitutes at most an 
exponentially small fraction (at most roughly $2^{-(c_2-c_1)n})$) of the total set of these assignments. The remainder of the section is devoted to the proof of this theorem.

\subsection{Lower bound on the number of satisfying assignments}

In this Section we quote relevant results from prior literature which give lower bounds on the cardinality of the satisfying set ${\rm SAT}(\bfPhi(n,m))$; that is, the number of satisfying assignments of our random formula $\bfPhi(n,m)$.

\begin{theorem}\label{theorem:sat-lower-bound}
Suppose $\beta< 2^K (\log 2)-\frac{(K+1)(\log 2)+3}{2}$ and $K\ge 3$. Then w.h.p.\ as $n\to\infty$ $\bfPhi(n,\beta n)$ is satisfiable. Furthermore,
w.h.p.\ as $n\to\infty$
\begin{align*}
n^{-1}\log |{\rm SAT}(\bfPhi(n,\beta n))|\ge \log 2+ (1/2)\beta \log\left(1-2^{1-K}+2^{-2K}-K2^{-K}\left(1-2^{-K}\right)\left(2^{1-K}+3K2^{-2K}\right)\right)
-o(1).
\end{align*}
\end{theorem}

This theorem can be found in~\cite{AchlioptasCojaOghlanRicciTersenghi} as Theorem 6 complemented with a lower bound (26).

Considering now the regime of interest to us, namely $\beta=\alpha 2^K\log 2, \alpha<1$, and using the Taylor expansion $\log(1+x)=x+O(x^2)$,
we can simplify the lower bound as  $\log 2-(\log 2)\alpha -O(1/2^K)$. We obtain the following corollary.

\begin{coro}\label{coro:sat-lower-bound}
For every $\alpha\in (0,1)$ and $\delta>0$, the following holds for all large enough $K$: w.h.p.\ as $n\to\infty$,
\begin{align}\label{eq:sat-lower-bound}
n^{-1}\log |{\rm SAT}(\bfPhi(n,\alpha 2^K (\log 2)n))|\ge (\log 2) (1-\alpha)-\delta.
\end{align}
\end{coro}

\subsection{Bounds on controlled first and second moments}

We fix $1/2>\epsilon>0$ and $\lambda>0$ (which are arbitrary for now), and $\alpha\in (0,1)$. Let $m=\alpha 2^K(\log 2) n$. We write $\bfPhi(n,m)$ as $\bfPhi(n)$ for short.
Recall that $Z(\bfPhi(n),\epsilon,\lambda n)=| {\rm SAT}(\bfPhi(n),\epsilon,\lambda n)|$.
 We now compute bounds on a ``controlled'' (appropriately defined) second moment of 
$Z(\bfPhi(n),\epsilon,\lambda n)$.

Fix $t\in \{0,1,\ldots,n\}$ and let ${\rm SAT}_2(\bfPhi(n),\epsilon,\lambda n,t)$ denote the set of ordered pairs $x^1,x^2\in {\rm SAT}(\bfPhi(n),\epsilon,\lambda n)$ such 
that $d(x^1,x^2)=n-t$. We let 
$Z_2(\bfPhi(n),\epsilon,\lambda n,t)=|{\rm SAT}_2(\bfPhi(n),\epsilon,\lambda n,t)|$. 
Note that $Z^2(\bfPhi(n),\epsilon,\lambda n)=\sum_t Z_2(\bfPhi(n),\epsilon,\lambda n,t)$. 

Fix $x^1,x^2\in\{0,1\}^n$, where $x^j$ has variable assignments $x_i^j, i\in [n]$ and $d(x^1,x^2)=n-t$. Fix $S_1,S_2\subset [n]$ with $|S_1|,|S_2|\ge (1-\epsilon)n$.
Let $S=S_1\cap S_2$.
Denote by $S(x^1,x^2)\subset S$ the set of coordinates  
in $S$ where $x^1$ and $x^2$ agree:
\begin{equation*}
    S(x^1,x^2)=\{i\in S: x^1_i=x^2_i\}.
\end{equation*}
Consider a clause $C$ in $\bfPhi$ conditioned on the event $C\in C(S_1)\cap C(S_2)=C(S)$. In particular, $C$ consists entirely of variables in $S$.
Let us compute the probability that $C$ is satisfied by both $x^1$ and $x^2$ when conditioned on this event.
We denote by $x\in {\rm SAT}(C)$ the event that string $x$ satisfies $C$. Let $E_C$ be the event that $C$ has distinct variables. Then (for sufficiently large $n$)
\begin{align*}
\pr(x^1,x^2 \in {\rm SAT}(C))&\leq 1-\pr\left(E_C\right)\left(\pr(x^1\notin {\rm SAT}(C)\mid E_C)+\pr(x^2 \notin {\rm SAT}(C)\mid E_C)-\pr(x^1,x^2 \notin {\rm SAT}(C)\mid E_C)\right)\\
&=1-\left(2/2^K-\pr(x^1,x^2 \notin {\rm SAT}(C))\right)\prod\limits_{i=1}^{K-1}\left(1-\frac{i}{n}\right)\\
&\leq 1-\left(1-O\left(\frac{K}{n}\right)\right)\left(2/2^K-\pr(x^1,x^2 \notin {\rm SAT}(C)\mid E_C)\right).
\end{align*}
We claim that 
\begin{align}\label{eq:Sx1x2-vsS}
\pr(x^1,x^2 \notin {\rm SAT}(C)\mid E_C)\leq 2^{-K}\left( |S(x^1,x^2)|/ |S|\right)^K.
\end{align}
Indeed, in order for $C$ to be violated by both $x^1$ and $x^2$ it has to be the case that $C$ consists entirely of variables
in $S(x^1,x^2)$, since otherwise it contains a variable $x_i$ where $x^1_i\ne x^2_i$ and at least one of $x^1,x^2$ satisfies $C$.
Conditioned on the event that $C$ consists entirely of variables in $S(x^1,x^2)$ the likelihood that it is violated (by both $x^1,x^2$) is $2^{-K}$. Similarly, the likelihood that $C$ consists entirely of variables
in $S(x^1,x^2)$ is upper-bounded $\left(|S(x^1,x^2)|/|S|\right)^K$. The claim follows.

We now obtain a bound on the likelihood that $x^1$ and $x^2$ belong to the set 
${\rm SAT}_2(\bfPhi(n),\epsilon,\lambda n,t)$ simultaneously.

\begin{lemma}\label{lemma:upper-SAT-2}
There exist functions $g,h:\R_+\to\R_+$ which depend on $\alpha,K$ satisfying $\lim_{z\to 0}g(z)=\lim_{z\to 0}h(z)=0$ such that for
any $x^1,x^2$ with $d(x_1,x_2)=n-t$,  for all large enough $n$ it holds that:
\begin{align*}
\pr\left(x^1,x^2 \in {\rm SAT}_2(\bfPhi(n),\epsilon,\lambda n,t) \right)
&\le 
n^2
\exp\left( g(\epsilon)n+h(\lambda) n\right)\left(1-\left(1-O\left(\frac{K}{n}\right)\right)\left(2/2^K-2^{-K}\left(\frac{t}{n}\right)^K\right)\right)^m.
\end{align*}
\end{lemma}

\begin{proof}
Fix $x^1,x^2$ and $S_j, |S_j|\ge (1-\epsilon) n, j=1,2$. As before let $S=S_1\cap S_2$. Then $|S|\ge (1-2\epsilon)n$
and $C(S)\subset C(S_j), j=1,2$. We assume that $x^j$ satisfies all but $\lambda n$ clauses in $C(S_j)$. Then
$x^1$ and $x^2$ satisfy at least $|C(S)|-2\lambda n$ clauses in the set of clauses $C(S)$.  We trivially have that $|S(x^1,x^2)|\le t$.
Fix $m_0\le m$ and fix any subset $\mathcal{C}\subset [m]$ with cardinality $m_0$.
Conditioned on the event $C(S)=\{C_j, j\in \mathcal{C}$\}, by our earlier claim and the union bound, the event $x^1,x^2\in {\rm SAT}(\bfPhi(n), C(S), \lambda n)$ 
occurs with probability at most
\begin{equation*}
\begin{aligned}
\sum_{r\le \lambda n}{\binom{m_0}{r}} &\left(1-\left(1-O\left(\frac{K}{n}\right)\right) \left(2/2^K+2^{-K}\left( |S(x^1,x^2)|/ |S|\right)^K\right)\right)^{m_0-r}\\
&\le n {\binom{m}{\lambda n}} \left(1-\left(1-O\left(\frac{K}{n}\right)\right)\left(2/2^K-2^{-K}\left( {\frac{t}{(1-2\epsilon)n}}\right)^K\right)\right)^{m_0-2\lambda n}.
\end{aligned}
\end{equation*}

The likelihood of $|C(S)|=m_0$ is
\begin{align*}
{\binom{m}{m_0}}\left({\frac{|S|}{n}} \right)^{Km_0} \left(1- \left({\frac{|S|}{n}}\right)^K \right)^{(m-m_0)}  
&\le
{\binom{m}{m_0}}\left({ \frac{|S|}{n}} \right)^{Km_0} 
\left(1- \left(1-2\epsilon\right)^K \right)^{(m-m_0)} \\
&\le 
{\binom{m}{m_0}}\left(1- \left(1-2\epsilon\right)^K\right)^{m-m_0}.
\end{align*} 
We obtain
\begin{align}
&\pr\left(x^1,x^2\in {\rm SAT}(\bfPhi(n), C(S), \lambda n) \right) \notag\\
&\le 
\sum_{m_0\le m}{\binom{m}{m_0}}\left(1- \left(1-2\epsilon\right)^K\right)^{m-m_0}
n  {\binom{m}{\lambda n}}\left(1+\left(1-O\left(\frac{K}{n}\right)\right)\left(2/2^K-2^{-K}\left( {\frac{t}{(1-2\epsilon)n}}\right)^K\right)\right)^{m_0-2\lambda n} \notag \\
&=
n {\binom{m}{\lambda n}} 
\left(1+(1- \left(1-2\epsilon\right)^K)-\left(1-O\left(\frac{K}{n}\right)\right)\left(2/2^K-2^{-K}\left( {\frac{t}{(1-2\epsilon)n}}\right)^K\right)\right)^{m}\notag \\
&\times\left(1-\left(1-O\left(\frac{K}{n}\right)\right)\left(2/2^K-2^{-K}\left( {\frac{t}{(1-2\epsilon)n}}\right)^K\right)\right)^{-2\lambda n}  \notag \\
&\le 
n {\binom{m}{\lambda n}} 
\left(1+(1- \left(1-2\epsilon\right)^K)-\left(1-O\left(\frac{K}{n}\right)\right)\left(2/2^K-2^{-K}\left( {\frac{t}{(1-2\epsilon)n}}\right)^K\right)\right)^{m}
\left(1-2^{-K+1}\right)^{-2\lambda n} . \label{eq:m-bound}
\end{align}
Taking now a union bound over the choices of $S$ we obtain an upper bound on 
$\pr\left(x^1,x^2 \in {\rm SAT}_2(\bfPhi(n),\epsilon,\lambda n,t) \right)$ which is (\ref{eq:m-bound}) multiplied by
$\sum_{(1-2\epsilon) n\le i\le n}{\binom{n}{i}}\le (2\epsilon n) {\binom{n}{2\epsilon n}}$. Thus:
\begin{align*}
&\pr\left(x^1,x^2 \in {\rm SAT}_2(\bfPhi(n),\epsilon,\lambda n,t) \right) \\
&\le 
 (2\epsilon n) {\binom{n}{2\epsilon n}}
 n {\binom{m}{\lambda n}} 
 \left(1+(1- \left(1-2\epsilon\right)^K)-\left(1-O\left(\frac{K}{n}\right)\right)\left(2/2^K-2^{-K}\left( {\frac{t}{(1-2\epsilon)n}}\right)^K\right)\right)^{m}
\left(1-2^{-K+1}\right)^{-2\lambda n} . 
\end{align*}
We recall that $m\le 2^K (\log 2) n=\Theta(n)$. We also have that: 
\begin{align*}
{\binom{n}{2\epsilon n}}\le\left({\frac{ne}{2\epsilon n}}\right)^{2\epsilon n}=\exp(2\epsilon \log(e/(2\epsilon)) n).
\end{align*}
We have that $g_0(\epsilon)\triangleq 2\epsilon \log(e/(2\epsilon)) \to 0$ as $\epsilon\to 0$. We treat similarly the other terms dependent on $\epsilon$
and $\lambda$. In particular, we
use that $1- \left(1-2\epsilon\right)^K\to 0$ as
$\epsilon\to 0$.
Finally, we use that $(2\epsilon n)n\le n^2$, and complete the proof.
\end{proof}

The cardinality of the set of pairs $x^1,x^2$ with $d(x^1,x^2)=n-t$ is $2^n{n\choose t}$. Applying Lemma~\ref{lemma:upper-SAT-2}
we obtain:
\begin{align*}
\E[Z_2(\bfPhi(n),\epsilon,\lambda n,t)]
\le 2^n{n\choose t}
n^2\exp\left( g(\epsilon)n+h(\lambda) n\right)\left(1-\left(1-O\left(\frac{K}{n}\right)\right)\left(2/2^K-2^{-K}\left({\frac{t}{n}}\right)^K\right)\right)^m
\end{align*}
for all sufficiently large $n$.
By allowing $t$ to be in a range of values of the form $a_1n\le t\le a_2 n$ we also obtain:
\begin{align}\label{eq:sum-Z2}
\E\left[\sum_{a_1n\le t\le a_2 n}Z_2(\bfPhi(n),\epsilon,\lambda n,t)\right]
\le &
n^3\exp\left( g(\epsilon)n+h(\lambda) n\right)\notag
\\
&\times\max_{a_1n\le t\le a_2 n}2^n{n\choose t}\left(1-\left(1-O\left(\frac{K}{n}\right)\right)\left(2/2^K-2^{-K}\left({\frac{t}{n}}\right)^K\right)\right)^m.
\end{align}
Here, the extra factor $n$ is associated with summing over $a_1n\le t\le a_2 n$ which is at most $(a_2-a_1)n\le n$.

We extract from the right-hand side the expression  under the max which  dominates the entire expression in exponential in $n$ terms.
In particular, let
\begin{align*}
z_2(n,t,K,\alpha)\triangleq 
2^n{\binom{n}{t}}\left(1-\left(1-O\left(\frac{K}{n}\right)\right)\left(2/2^K-2^{-K}(t/n)^K\right)\right)^m,
\end{align*}
so that the expression under the max is $\max_{a_1n\le t\le a_2 n}z_2(n,t,K,\alpha)$.
Using ${n\choose t}=\exp( n H(t/n)+o(n))$ where $H(x)=-x\log x-(1-x)\log(1-x)$ is the binary entropy function, we obtain:
\begin{align*}
n^{-1}\log z_2(n,t,K,\alpha)
=
\log(2)+H(t/n)+(m/n)\log\left(1-\left(1-O\left(\frac{K}{n}\right)\right)\left(2/2^K-2^{-K}\left(t/n\right)^K\right)\right)+o(1).
\end{align*}
Recall that $m=\alpha 2^K(\log 2) n, \alpha<1$.
We also write $t=sn$. We obtain after Taylor expanding the $\log$ function that
\begin{align*}
n^{-1}\log z_2(n,t,K,\alpha)&=\log(2)+H(s)+\alpha(\log 2) 2^K\log\left(1-\left(1-O\left(\frac{K}{n}\right)\right)\left(2/2^K-2^{-K} s^K\right)\right) +o(1) \\
&=\log(2)+H(s)-2(\log 2) \alpha+(\log 2)\alpha s^K+O_K\left(2^{-K}\right) +o(1) \\
&\triangleq C(\alpha,s,K)+o(1). 
\end{align*}

\begin{lemma}\label{lemma:rate-exponent}
There exist $0<\nu_1<2\nu_1<\nu_2<1/2$  such that for every $\alpha\in (1/2+1/10,1)$  the following holds.
\begin{align}
& \lim_{K\to\infty}\sup_{s\in [1-\nu_1,1]}C(\alpha,s,K) = \lim_{K\to\infty}C(\alpha,1,K)=(\log 2)(1-\alpha)>0,   \label{eq:[nu,1]} \\
& \lim_{K\to\infty}\sup_{s\in [1-\nu_2,1-\nu_1]}C(\alpha,s,K) \le (\log 2)(1/2-\alpha)\le -(\log 2)/10<0.  \label{eq:[nu,nu]}
\end{align}
\end{lemma}
We note that  the first  identity regards only $\nu_1$. The choice of $\alpha>1/2+1/10$ is somewhat arbitrary and can be \change{replaced}
by any constant $1>c>1/2$ which does not depend on $K$.
\begin{proof}
Consider the part of $C(\alpha,s,K)$ which depends on $s$, namely, $H(s)+(\log 2)\alpha s^K$.
We consider separately the cases $s\le 1-\log K/K$ and $s> 1-\log K/K$. In the first case we note that $s^K=O(1/K)\to 0$ as $K\to\infty$.
To prove (\ref{eq:[nu,1]}),  by taking $\nu_1$ sufficiently close to $0$,  we obtain that 
$\max_{1-\nu_1\le s\le 1}H(s)$ is sufficiently close to $0$ and in particular is at most $(1/2)\log 2<(\log 2)\alpha$. 
 As a result, in this case:
\begin{align*}
\lim_{K \to \infty}\sup_{s\in [1-\nu_1,1-\log K/K]}C(\alpha,s,K) &\le \log(2)+(1/2)\log 2-2(\log 2) \alpha  \\
&< (\log 2)(1-\alpha).
\end{align*}
When $s>1-\log K/K$ we have that $(\log 2)\alpha s^K\le (\log 2)\alpha$ and $H(s)\to 0$ as $K\to\infty$, and thus (\ref{eq:[nu,1]}) is obtained
by taking $s=1$.

For (\ref{eq:[nu,nu]}) we note that for any fixed $0<\nu_1<2\nu_1<\nu_2<1/2$ we have that $\lim_K \sup_{1-\nu_2\le s\le 1-\nu_1}s^K=0$ and thus
\begin{align*}
\lim_{K\to\infty}\sup_{s\in [1-\nu_2,1-\nu_1]}C(\alpha,s,K) &\le \log(2)+H(1-\nu_2)-2(\log 2) \alpha \\
&= 2(\log 2)(1/2-\alpha)+H(1-\nu_2)\\
&\le -2(\log 2)/10+H(1-\nu_2).
\end{align*}
Taking $\nu_2$ sufficiently close to $0$  so that $ H(1-\nu_2)<(\log 2)/10$ 
we obtain (\ref{eq:[nu,nu]}). 
\end{proof}

We now use these  bounds for (\ref{eq:sum-Z2}) to arrive at the following corollary:
\begin{coro}\label{coro:trim}
There exist $0<\nu_1<2\nu_1<\nu_2<1/2$ such that for every $\alpha\in (1/2+1/10,1)$ and $\delta>0$  the following holds. For every  sufficiently large $K$
there exist $\bar\epsilon,\bar\lambda>0$ such that for all $\epsilon\le \bar\epsilon,\lambda\le \bar\lambda$, for all sufficiently large $n$,
\begin{align*}
&\limsup_{n\to\infty} n^{-1}\log \E\left[\sum_{t\in  [(1-\change{\nu_1}) n,n]}Z_2(\bfPhi(n),\epsilon,\lambda n,t) \right] 
\le (\log 2)(1-\alpha)+\delta, \\
&\limsup_{n\to\infty} n^{-1}\log \E\left[\sum_{t\in  [(1-\nu_2) n, (1-\nu_1) n]}Z_2(\bfPhi(n),\epsilon,\lambda n,t)\right] \le -(\log 2)/(40).
\end{align*}
\end{coro}
We note that $\delta$ appears only in the first identity. The reasons for allowing $\delta$ to be arbitrarily small is to allow the bound to hold for any 
choice of $\alpha$ in the stated range. This will be necessary for the proof of Theorem~\ref{theorem:clustering} below.

\begin{proof}
We fix $\nu_1,\nu_2$ as in Lemma~\ref{lemma:rate-exponent}. Let  $\alpha\in (1/2+1/10,1)$ and let $\delta>0$ be arbitrary.
We find $K_0$ sufficiently large so that for all $K\ge K_0$,
\begin{align*}
& \sup_{s\in [1-\nu_1,1]}C(\alpha,s,K) = 
C(\alpha,1,K)
\le
(\log 2)(1-\alpha)+\delta/2,    \\
& \sup_{s\in [1-\nu_2,1-\nu_1]}C(\alpha,s,K)
\le -(\log 2)/20<0.
\end{align*}
For every $K\ge K_0$ we find $\bar\epsilon,\bar\lambda$ sufficiently small so that $g(\epsilon)+h(\lambda)<\min\left(\delta/2,\log(2)/40\right)$ for all 
$\epsilon\le\bar\epsilon, \lambda\le \bar \lambda$. 
The required bounds are obtained from (\ref{eq:sum-Z2})  and using the fact that $-(\log 2)/20+2(\log 2)/40=-(\log 2)/40$.
\end{proof}

\subsection{Proof of Theorem~\ref{theorem:clustering}}
We are now ready to prove Theorem~\ref{theorem:clustering}.

\begin{proof}[Proof of Theorem~\ref{theorem:clustering}]
We fix $\nu_1,\nu_2$ as in the setting of Corollary~\ref{coro:trim}. Similarly, we fix any $\alpha\in (1/2+1/5,1)$ and any sufficiently small $\delta>0$ such that
\begin{align}\label{eq:delta}
\delta \leq \frac{1}{7} (\log 2)(1-\alpha).
\end{align}
We find large enough $K$ and small enough $\bar\epsilon, \bar\lambda$ so that Corollary~\ref{coro:trim} applies. 
Furthermore, we assume that $K$ is large enough so that the bound in Corollary~\ref{coro:sat-lower-bound} applies as well. We let $\epsilon<\bar\epsilon$ and $\lambda<\bar\lambda$ be arbitrary.
By Markov's inequality applied to the second bound in Corollary~\ref{coro:trim}  and by lowering the constant to $-(\log 2)/50$ we obtain that for all large enough $n$:
\begin{align*}
\pr\left(\sum_{t\in  [(1-\nu_2) n, (1-\nu_1) n]} Z_2(\bfPhi(n),\epsilon,\lambda n,t) \ge 1\right) \le \exp\left(-n(\log 2)/50\right),
\end{align*}
and thus $Z_2(\bfPhi(n),\epsilon,\lambda n,t)=0$ for all $t\in  [(1-\nu_2) n, (1-\nu_1) n]$ w.h.p.\ as $n\to\infty$. This means
that w.h.p.\ as $n\to\infty$ for \change{every 
two subsets} $S_j, j=1,2$ such that $|S_j|\ge (1-\epsilon) n$ and every 
two solutions $x_j, j=1,2$
 such that $x_j$ violates at most $\lambda n$ clauses
out of the clauses $C(S_j)$---that is, clauses consisting entirely of variables in  $S_j$---it is the case that either $n^{-1}d(x_1,x_2)\le \nu_1$ or $n^{-1}d(x_1,x_2)\ge \nu_2$. 
In particular, the set  ${\rm SAT}(\bfPhi(n),\epsilon,\lambda n)$ exhibits the $(\nu_1,\nu_2)$-OGP.

Consider the implied unique $(\nu_1,\nu_2)$-clustering $CL_\ell, \ell\ge 1$ of this set. 
Note that  we trivially have that
\begin{align*}
\bigcup_\ell CL_\ell = {\rm SAT}(\bfPhi(n),\epsilon,\lambda n)\supset {\rm SAT}(\bfPhi(n,\alpha 2^K(\log 2)n)).
\end{align*} 
By Corollary~\ref{coro:sat-lower-bound} w.h.p.\ as $n\to\infty$:
\begin{align}\label{eq:sat-lower-bound-2}
|{\rm SAT}(\bfPhi(n),\epsilon,\lambda n) |
\ge 
|{\rm SAT}(\bfPhi(n,\alpha 2^K(\log 2)n))|\ge \exp\left( n (\log 2) (1-\alpha)-n\delta\right).
\end{align}
Note that
\begin{equation*}
\frac{1}{4}\left(\max_\ell |CL_\ell | \right)^2\leq\sum\limits_\ell\binom{|CL_\ell|}{2}=\sum_{t\in  [(1-\nu_1) n,n]} Z_2(\bfPhi(n),\epsilon,\lambda n,t).
\end{equation*}
By the first part of Corollary~\ref{coro:trim}, applying Markov's inequality we have that w.h.p.\ as $n\to\infty$:
\begin{align*}
\sum_{t\in  [(1-\nu_1) n,n]} Z_2(\bfPhi(n),\epsilon,\lambda n,t)  \le \exp\left( (\log 2)(1-\alpha)n+2\delta n \right).
\end{align*}
Then we obtain that w.h.p. as $n\to\infty$\ 
\begin{align*}
\max_\ell |CL_\ell|\le 2\exp\left(\frac12 (\log 2)(1-\alpha) n +\delta n\right) \le \exp\left(\frac12 (\log 2)(1-\alpha) n +2\delta n  
\right).
\end{align*}
We set $c_1=\frac12(\log 2)(1-\alpha)+2\delta$, $c_2=(\log 2)(1-\alpha)-\delta$ and the proof is complete by our choice of $\delta$ in (\ref{eq:delta}) leading to $0<c_1<c_2$.
\end{proof}

\section{Proof of Combinatorial NLTS}\label{section:main-result}

In this section we give the proof of  Theorem~\ref{theorem:MainResult}.

We first establish a simple upper tail bound on  $m-| C(S) |$ when $S\subset [n]$ satisfies $|S|\ge (1-\epsilon) n$. Recall that $C(S)$ is the set of clauses in an instance $\Phi$ such that each clause in $C(S)$ consists entirely of variables
$x_i$ with  $i\in S$.
\begin{prop}\label{prop:eta}
For every $\alpha\le 1$, $K\ge 1$, $\eta>0$ and $m=\alpha 2^K(\log 2) n$, there exists small enough $\epsilon>0$ so that w.h.p.\ as $n\to\infty$
the formula $\bfPhi(n,m)$ satisfies
$m-|C(S)|\le \eta n$ for all $S\subset [n], |S|=(1-\epsilon)n$.
\end{prop}
\begin{proof}
 Fix $S\subset [n]$. The probability
that a random clause belongs to $C(S)$ (namely, consists entirely of variables $x_i, i\in S$) is  $(|S|/n)^K$.
Then
\begin{align*}
\E[m-|C(S)|]=m\left(1- (|S|/n)^K\right).
\end{align*}

\change{Note that $|C(S)|$ satisfies the bounded difference property that changing the clause
choice from $C_j$ to a different clause $C_j'$ changes $|C(S)|$ by at most $1$}. Thus, applying \change{McDiarmid's inequality to bound the probability that $m-|C(S)|$ exceeds its expectation by $t$},
we obtain that for every $t\ge 0$
\begin{align*}
\pr\left( m-|C(S)|\ge \E[m-|C(S)|]+t\right)\le \exp\left(- \frac{2t^2}{m}\right).
\end{align*}
In particular,
\begin{align*}
\pr\left( m-|C(S)|\ge n\eta\right)\le \exp\left(-2 \frac{\left(n\eta - m\left(1- (|S|/n)^K\right)\right)^2}{m}\right).
\end{align*}
Assume that $\epsilon>0$ is small enough so that $\eta-2^K(1-(1-\epsilon)^K)>0$.
Assuming $S$ satisfies $|S|=(1-\epsilon) n$ for this choice of $\epsilon$ and using $m\le 2^K n$ we obtain:
\begin{align*}
\pr\left( m-|C(S)|\ge \eta\right) & \le \exp\left(-2 {\left(n\eta - m\left(1- (1-\epsilon)^K\right)\right)^2 \over m}\right) \\
&\le
 \exp\left(- {\left(n\eta - 2^K n\left(1- (1-\epsilon)^K\right)\right)^2 \over 2^{K-1}n}\right) \\
 &=\exp\left(- n{\left(\eta - 2^K \left(1- (1-\epsilon)^K\right)\right)^2 \over 2^{K-1}}\right).
\end{align*}
By the union bound,
\begin{align*}
\pr\left( \exists S\subset [n]: |S|= (1-\epsilon)n, m-|C(S)|\ge \eta\right) 
&\le
{n\choose \epsilon n}
 \exp\left(- n{\left(\eta - 2^K \left(1- (1-\epsilon)^K\right)\right)^2 \over 2^{K-1}}\right).
\end{align*}
We have that
\begin{align*}
{n\choose \epsilon n}\le (ne/(\epsilon n))^{\epsilon n}=\exp( n \epsilon \log(e/\epsilon)).
\end{align*}
As $\epsilon\to 0$ we have that $\epsilon \log(e/\epsilon)\to 0$.  Similarly, $1- (1-\epsilon)^K\to 0$ as $\epsilon\to 0$. The upper
bound is then $\exp(-\Theta(n))$, and the proof is complete.
\end{proof}

We fix for now an arbitrary instance of the formula $\Phi(n,m)$. We only assume that it is satisfiable: ${\rm SAT}(\Phi(n,m))\ne\emptyset$.
Fix $\epsilon>0$ and let $\eta=\eta(\Phi,S)$ be defined by:
\begin{align}\label{eq:eta}
\eta\triangleq {1\over n}\max_{S \subseteq [n]:|S|= (1-\epsilon) n} (m-|C(S)|).
\end{align}
Namely, $n\eta$ is the largest number of clauses across all choices of $S, |S|=(1-\epsilon)n$ containing at least one variable  in $S^c$.

We start with some preliminary derivations.
We introduce a (somewhat noncanonical) basis of $\C^{2^{\change{N}}}$. For any $r>0$ and any $\sigma \in \{0,1\}^r$, denote by $\bar\sigma\in \{0,1\}^r$
a bit string obtained by flipping elements of $\sigma$ to the opposite element (e.g., for $\sigma=(0,0,1,0)$, we have that $\bar\sigma=(1,1,0,1)$).
The set of elements of our basis are defined by:
\begin{align}
\left\{ 
\bigotimes_{i\in [n]} \left( {1\over \sqrt{2}} |\sigma_i \rangle \pm{1\over \sqrt{2}} |\bar \sigma_i\rangle \right) : \sigma_i \in \{0,1\}^{|D(i)|}, i \in [n] \label{eq:basis}
\right\},
\end{align}
where again $D(i)$ is the subset of qubits corresponding to the variable $x_i$. Here, to avoid ambiguity due to ${\bar{ \bar \sigma}}=\sigma$, we assume that $\sigma$ is smaller than $\bar\sigma$ in lexicographic order.
Note that the cat state $\ket{\rmCAT}=\frac{1}{\sqrt{2}}(\ket{0}^{\otimes \change{N}}+\ket{1}^{\otimes \change{N}})$ is an element of the basis (\ref{eq:basis}). For each $i$ the operator $(I-\ket{\rmCAT(i)}\bra{\rmCAT(i)})$
is diagonal in this basis, and for each element $|w\rangle $ of this basis we have that
\begin{align}\label{eq:w-cat}
\left(I-\ket{\rmCAT(i)}\bra{\rmCAT(i)}\right) |w\rangle =0
\end{align}
if the $i$-th element of the representation of $w$ in (\ref{eq:basis}) is exactly $\ket{\rmCAT(i)}$ and 
\begin{align}\label{eq:w-non-cat}
(I-\ket{\rmCAT(i)}\bra{\rmCAT(i)}) |w\rangle =|w\rangle
\end{align}
otherwise.

Let $\left\{|z\rangle\right\}, z\in \{0,1\}^{\change{N}}$ denote the computational basis. Each element of the basis (\ref{eq:basis}) expands in the computational basis
in a straightforward way as follows. For each element of the basis 
\begin{align*}
|w\rangle=\bigotimes_{i\in [n]} \left( {1\over \sqrt{2}} |\sigma_i \rangle \pm{1\over \sqrt{2}} |\bar \sigma_i\rangle \right)
\end{align*}
let $B(w)$ be the set of bit strings $z\in \{0,1\}^{\change{N}}$ such 
that for each subset $D(i), i\in [n]$ it holds that $z_{D(i)}=\sigma_i$ or $=\bar\sigma_i$. Then
\begin{align}
|w\rangle =  \left( {1\over \sqrt{2}} \right)^{n} \sum_{z\in B(w)}  (-1)^{R(w,z)} |z\rangle,   \label{eq:w-in-z}
\end{align}
where $R(w,z)$ is the number of $-1$-s appearing in the projection of $|w\rangle$ onto $|z\rangle$. 

We now characterize the $\epsilon$-near ground states in terms of the elements of this basis.

\begin{lemma}\label{lemma:basis-CAT}
Suppose $|\psi\rangle$ is an $\epsilon$-near  ground state for $H$ and $S\subset [n]$ is (any) associated subset. 
Define the (unnormalized) state $|\phi\rangle \triangleq Q^{-1}(\gamma) |\psi\rangle$.
Suppose $|w\rangle $ is an element of the basis such that for some $i\in S$, the $i$-th element of $|w\rangle $ in the representation (\ref{eq:basis}) is 
distinct from $\ket{\rmCAT(i)}$. Then $\langle \phi | w\rangle =0$. 
\end{lemma}
Namely, the expansion of $|\phi\rangle$ in this basis consists only of elements $|w\rangle$ of the form 
$\bigotimes_{i\in [n]} \left( {1\over \sqrt{2}} |\sigma_i \rangle \pm{1\over \sqrt{2}} |\bar \sigma_i\rangle \right)$ where 
${1\over \sqrt{2}} |\sigma_i \rangle \pm{1\over \sqrt{2}} |\bar \sigma_i\rangle=\ket{\rmCAT(i)}$ for each $i\in S$.

\begin{proof}
For every $i\in S$ we have that
\begin{align*}
0&=\langle \psi | H_i |\psi\rangle \\
&=\langle \phi| Q(\gamma) H_i Q(\gamma) |\phi\rangle\\
& = \langle \phi| Q(\gamma) Q_{x_i,\gamma}^{-1}\left (I-\ket{\rmCAT(i)}\bra{\rmCAT(i)}\right) Q_{x_i,\gamma}^{-1} Q(\gamma) |\phi\rangle \\
&=\langle \phi| \left(\bigotimes_{j\notin C(x_i)}  Q_{C_j,\gamma}\right)\left (I-\ket{\rmCAT(i)}\bra{\rmCAT(i)}\right) \left(\bigotimes_{j\notin C(x_i)}  Q_{C_j,\gamma}\right)  |\phi\rangle,
\end{align*}
where $\bigotimes_{j\notin C(x_i)}  Q_{C_j,\gamma} $ is assumed to act as identity operator on qubits associated with clauses in $C(x_i)$,
namely qubits in the set $\{(j,k), j\in C(x_i), k\in [K]\}$. Let $|\tilde \psi_i\rangle \triangleq \bigotimes_{j\notin C(x_i)}  Q_{C_j,\gamma}  |\phi\rangle$. 
Then:
\begin{align*}
0=
\langle \tilde \psi_i| \left (I-\ket{\rmCAT(i)}\bra{\rmCAT(i)}\right)  |\tilde \psi_i\rangle.
\end{align*}

We expand $|\tilde\psi_i\rangle$ in the basis 
$|\tilde\psi_i\rangle =\sum_w \langle w | \tilde\psi_i \rangle |w\rangle$ previously introduced to obtain: 
\begin{align*}
0= \sum_{w_1,w_2}\langle \tilde \psi_i| w_1\rangle \langle  w_2 | \tilde\psi_i\rangle 
\langle w_1  |\left (I-\ket{\rmCAT(i)}\bra{\rmCAT(i)}\right)    |w_2\rangle.
\end{align*}
Applying (\ref{eq:w-cat}) and (\ref{eq:w-non-cat}) we obtain:
\begin{align*}
0=\sum_{w\in Y_i}|\langle |\tilde\psi_i | w\rangle|^2,
\end{align*}
where the sum is over the subset $Y_i$ of elements $|w\rangle$ of the basis such that its $i$-th element in the representation (\ref{eq:basis}) is distinct from $\ket{\rmCAT(i)}$.
We thus obtain that $\langle \tilde\psi_i | w\rangle=0$ for all $\ket{w}\in Y_i$ and thus the expansion of $\tilde \psi_i$ in this basis
contains only elements $|w\rangle$ with the $i$-the element equal to $\ket{\rmCAT(i)}$. 

Recall that $|\tilde \psi_i\rangle = \left(\bigotimes_{j\notin C(x_i)}  Q_{C_j,\gamma}\right) |\phi\rangle$. The operator $\bigotimes_{j\notin C(x_i)}  Q_{C_j,\gamma}$
 acts as the identity operator on qubits in $D(i)$. Then  
$\langle \phi | w\rangle =0$ if the $i$-th element of $|w\rangle$ is distinct from $\ket{\rmCAT(i)}$.
Since this applies to all $i\in S$ the proof is complete.
\end{proof}

Fix an $\epsilon$-near ground state $|\psi\rangle =Q(\gamma)|\phi\rangle$ and any associated subset $S\subset [n]$.
Let $W(S)$ be the subset of the basis consisting of $|w\rangle$ such that for each $i\in S$, 
the $i$-th element of $|w \rangle$ is $\ket{\rmCAT(i)}$ in our representation (\ref{eq:basis}). 
From (\ref{eq:eta}) we have that:
\begin{align}\label{eq:W(S)-upper}
|W(S)| \le 2^{n K\eta}
\end{align}
since $|S|\ge (1-\epsilon)n$ and each element of $D(i), i\notin S$ appears in at least one clause outside of $C(S)$.  
The implication of Lemma~\ref{lemma:basis-CAT} and (\ref{eq:w-in-z}) taken together is that $|\phi\rangle$ as in Lemma~\ref{lemma:basis-CAT} can be written as:
\begin{align*}
|\phi\rangle &= \sum_{|w\rangle \in W(S)} \braket{w | \phi} |w\rangle \\
&= \left( {1\over \sqrt{2}} \right)^{n}\sum_{|w\rangle \in W(S)} \braket{w | \phi} \sum_{z\in B(w)} (-1)^{R(w,z)} |z\rangle.
\end{align*} 
This implies that for every element of the computational basis $|z\rangle$ we have that:
\begin{align}
\langle \phi | z\rangle 
=
\left( {1\over \sqrt{2}} \right)^{n}\sum_{w: z\in B(w),|w\rangle \in W(S)} \braket{w | \phi}  (-1)^{R(w,z)}, \label{eq:phi-in-z}
\end{align}
with the understanding that the value is zero if the sum is empty. From this, using (\ref{eq:W(S)-upper}), and 
the Cauchy--Schwartz inequality $ | \langle \phi | w\rangle|\le \|\phi\|_2$, we obtain an upper bound for every $z$:
\begin{align}\label{eq:upper-weight}
|\langle \phi | z\rangle|
\le 
\left( {1\over \sqrt{2}} \right)^{n} \| \phi\|_2 |W(S)| \le \left( {1\over \sqrt{2}} \right)^{n} \| \phi\|_2  2^{K\eta n}.
\end{align}

Let $\bar S\subset \{0,1\}^{\change{N}}$ consist of bit strings $z\in \{0,1\}^{\change{N}}$ which are consistent on $S$, namely such that for each $i\in S$
the values of $z$ on $D(i)$ are all identically $0$ or identically $1$. We have that:
\begin{align}
|\bar S| = 2^{(1-\epsilon)n}|W(S)| \le 2^{n+K\eta n}. \label{eq:size-bar-S}
\end{align}

Any $z\in \bar S$  can be projected to a partial assignment of variables $x_i, i\in S$ in a natural way. 
We denote by $x_S(z)\in \{0,1\}^S$ this projection. Since every clause $C$ in $C(S)$ consists entirely of variables in $x_i, i\in S$, for every such
clause $C$ and for every $z \in \bar S$, this clause is either satisfied or violated by the projection $x_S(z)$, regardless of the values of 
the value of the variables $x_i, i\notin S$. Denote by $\ell(z)\ge 0$ the number of clauses in $C(S)$ violated by the projection $x_S(z)$. Note that even though $z\in \{0,1\}^{\change{N}}$
might not represent a consistent assignment of all variables $x_i$, for each clause $C_j$ it is still well defined whether $C_j$
is satisfied by $z$ or not.
We decompose $\bar S$ as $\bigcup_{r\ge 0}\bar S(r)$, where $\bar S(r)$ is the set of strings  $z\in \bar S$ with $\ell(z)=r$.

\begin{lemma}\label{lemma:bounds-in-l(z)}
The following bounds hold for every $r\ge 0$ and $z\in \bar S(r)$:
\begin{align}
\gamma^{r+\eta n} |\langle \phi | z\rangle| \label{eq:violation-weights}
\le
|\langle \psi | z\rangle|
 \le
  \gamma^{r} |\langle \phi | z\rangle|.
\end{align}
\end{lemma}

\begin{proof}
$Q(\gamma)$ is diagonal in the computational basis. For every $z\in \bar S$ we have that $Q(\gamma) | z\rangle =\gamma^{v} |z\rangle$, 
where $v$ is the total number of clauses in $\{C_j, j\in [m]\}$ violated by $z$\change{, including here clauses that contain classical variables $x_i$ for which $z$ is inconsistent}. Thus:
\begin{align*}
|\langle \psi | z\rangle| = |\langle \phi |Q(\gamma) | z\rangle|=\gamma^v |\langle \phi | z\rangle|.
\end{align*}
Note that when $z\in \bar S(r)$ we have by (\ref{eq:eta}) that  $r\le v\le r+\eta n$. 
Then (\ref{eq:violation-weights}) follows.
\end{proof}

For every $r\ge 0$ and every $z\in \bar S(r)$, by applying Lemma~\ref{lemma:bounds-in-l(z)} to (\ref{eq:upper-weight}) we obtain:
\begin{align}
|\langle \psi | z\rangle|
\le \gamma^r\left( {1\over \sqrt{2}} \right)^{n} \| \phi\|_2  2^{K\eta n}. \label{eq:z-in-S(r)}
\end{align}

We now obtain a complementary lower bound. Fix any two bit strings $z_1,z_2\in \bar S$ which agree on the set of qubits $\bigcup_{i\notin S}D(i)$.
Note then that for every $|w\rangle\in W(S)$, $z_1\in B(w)$ iff $z_2\in B(w)$.
Furthermore,
we claim that for every $|w\rangle \in W(S)$
\begin{align*} 
R(w,z_1)=R(w,z_2).
\end{align*}
Indeed, all minuses contributing to $R(w,z)$ arise from minuses in 
$w=\bigotimes_{i\in [n]} \left( {1\over \sqrt{2}} |\sigma_i \rangle \pm{1\over \sqrt{2}} |\bar \sigma_i\rangle \right)$ associated with $i\notin S$
by the definition of $W(S)$, and those are identical for $z_1$ and $z_2$ since they are identical outside of $\bigcup_{i\in S}D(i)$.
We then obtain from (\ref{eq:phi-in-z}) that $\langle \phi | z_1\rangle=\langle \phi | z_2\rangle$.

 Let $z^*$ be any maximizer
$z^*=\argmax_{z\in \{0,1\}^{\change{N}}} |\langle\phi|z\rangle|$. The maximization can be restricted to $z\in \bar S$ 
since by (\ref{eq:phi-in-z}) it has to be the case that $z\in B(w)$ for some $w\in W(S)$.
We thus trivially have from (\ref{eq:size-bar-S}) that:
\begin{align*}
 |\langle\phi|z^\ast\rangle|^2\ge 2^{-(1+K\eta)n}\|\phi\|_2^2. 
 \end{align*}
 
Recall that $\bar S(0)\subset \bar S$ consists of bit strings $z$ such that their projection $x_S(z)$ satisfies all clauses $C_j\in C(S)$. 
Namely, it is the set of bit strings in $\bar S$ such that $\ell(z)=0$.
For every $z\in \bar S(0)$, let $z^{z^*}$ be obtained from $z$ by changing coordinates of $z$ in $\bigcup_{i\notin S}D(i)$ to the one of $z^*$
and leaving them intact otherwise. By the above, $\langle \phi | z^{z^*}\rangle=\langle \phi | z^* \rangle$.
Note that for any two strings $z_1,z_2 \in \bar S(0)$ which differ on $\bigcup_{i\in S}D(i)$ we have that $z^{z^*}_1\ne z^{z^*}_2$.
We therefore obtain:
\begin{align*}
\sum_{z\in \bar S(0)} | \langle \phi | z^{z^*}\rangle |^2 \ge |\bar S(0)|2^{-(1+K\eta)n}\|\phi\|_2^2.
\end{align*}
Applying the lower bound part of Lemma~\ref{lemma:bounds-in-l(z)}
we obtain the lower bound:
\begin{align*}
\sum_{z\in \bar S(0)} | \langle \psi | z^{z^*}\rangle |^2 \ge |\bar S(0)|\gamma^{2\eta n}
2^{-(1+K\eta)n}\|\phi\|_2^2.
\end{align*}
Combining this as an upper bound on $\|\phi\|_2^2$ 
with (\ref{eq:z-in-S(r)}) we obtain the following estimate on measurement probabilities  $| \langle \psi | z\rangle |^2$
for $z\in \bar S(r)$:
\begin{align*}
| \langle \psi | z\rangle |^2
&\le 
{ \gamma^{2r}\left( {1\over 2} \right)^{n}  2^{2K\eta n} \over |\bar S(0)|\gamma^{2\eta n} 2^{-(1+K\eta)n}}
\sum_{z\in \bar S(0)} | \langle \psi | z^{z^*}\rangle |^2  \\
&=
{ \gamma^{2r-2\eta n}  2^{3K\eta  n} \over |\bar S(0)|}
\sum_{z\in \bar S(0)} | \langle \psi | z^{z^*}\rangle |^2 \\
&\le 
{ \gamma^{2r-2\eta n}  2^{3K\eta  n} \over |\bar S(0)|} \\
&\le 
{ \gamma^{2r}  \left({2 \over \gamma}\right)^{3K \eta n} \over |\bar S(0)|},
\end{align*}
where the third bound uses $\|\psi\|_2=1$. This an important finding which we summarize as follows.

\begin{theorem}\label{theorem:probability-bounds}
For every $r\ge 0$ and $z\in \bar S(r)$ 
\begin{align*}
| \langle \psi | z\rangle |^2 \le { \gamma^{2r}  \left({2 \over \gamma}\right)^{3K\eta n} \over |\bar S(0)|}.
\end{align*}
\end{theorem}
Loosely speaking, the bound in the theorem states that the probability weight associated with strings $z$ violating $r$ clauses is
at most roughly $\gamma^{2r}/|\bar S(0)|$. We will argue that the cardinality of $\bar S(0)$ is exponentially large with a controlled exponent.
By making $r$ linear in $n$, then, the contribution of these $r$-violating bit strings will be negligible, and thus most of the probability mass of $\lvert\langle z|\psi\rangle\rvert^2$
is spread over the nearly satisfying solutions. Furthermore, the probability weight associated with each individual string in $\bar S(0)$ is also 
roughly at most $1/|\bar S(0)|$. Namely, $\lvert\langle z|\psi\rangle\rvert^2$ is roughly uniform over nearly satisfying solutions.

From our robust OGP, the distribution of $\ket{\psi}$ over nearly satisfying solution is thus well-spread in the sense that it has large mass over clusters that are far apart. To translate this into a lower bound on circuit depth we will use the Lemma below, variants of which can be found throughout the literature~\cite{eldar2017local,moosavian2022limits,anshu2022nlts}.
\begin{lemma}[Fact 4 of \cite{anshu2022nlts}] \label{lem:circuit_lower_bound}
    Let $p$ be a distribution on $n$ bits generated by measuring the output of a quantum circuit acting on the initial state $\ket{0}^{\otimes n}$. If there exist sets $S_1,S_2 \subseteq \{0,1\}^n$ where the Hamming distance between elements of $S_1$ and $S_2$ is at least $d(S_1, S_2)$ and $p(S_1) \geq \mu$ and $p(S_2) \geq \mu$, then the depth $T$ of the circuit must be at least
    \begin{equation}
        T \geq \frac{1}{3}\log \left( \frac{d(S_1, S_2)^2}{400n \log(1/\mu) } \right).
    \end{equation}
\end{lemma}

\subsection{Proof of Theorem~\ref{theorem:MainResult}}

\begin{proof}[Proof of Theorem~\ref{theorem:MainResult}]
We assume the setting of Theorem~\ref{theorem:clustering}. Our formula $\Phi$ is $\bfPhi(n,m)$. 
The parameters $\alpha,K$, etc.\ are as in the theorem. 

We select $\epsilon \le \bar\epsilon, \lambda\le \bar\lambda$
and $\gamma$ as follows. We fix any $0<\lambda\le \bar \lambda$. We select $\gamma$ small enough so that
so that $\gamma^{2\lambda}  <1/8$. We select $\eta>0$ small enough so that
\begin{align}
\left({2 \over \gamma}\right)^{4K\eta } 
&\le 
2^{{c_2-c_1 \over 2}}  \label{eq:parameters1} \\
\gamma^{2\lambda -3K\eta } & <1/8. \label{eq:parameters2} \\
4K\eta & <1.  \label{eq:parameters3}
\end{align}
We finally select $\epsilon<\bar \epsilon$ also small enough with respect to $K$ and $\eta$ such that Proposition~\ref{prop:eta} applies.

Assume that the high probability event underlying Theorem~\ref{theorem:clustering} and Proposition~\ref{prop:eta} takes place for this choice of parameters.
Denote by $CL_\ell$ the associated $(\nu_1,\nu_2)$-clustering of the set ${\rm SAT}(\bfPhi(n, \epsilon,\lambda n))$. 
In particular, $|CL_\ell |\le 2^{c_1 n}$ for all $\ell$.
Fix an $\epsilon$-near ground state
$|\psi\rangle$ and any associated subset $S\subset [n], |S|\ge (1-\epsilon)n$.
 In particular, $\bar S(0)\ne\emptyset$, and furthermore by the conclusion of Theorem~\ref{theorem:clustering}:
\begin{align}
|\bar S(0)|\ge 2^{c_2 n}. \label{eq:barSzeroLower}
\end{align}
We claim the following relations: 
\begin{align*}
\sum_{r> \lambda n}\sum_{z\in \bar S(r)}| \langle \psi | z\rangle |^2 
&\le 
O(n) { \gamma^{2\lambda n}  \left({2 \over \gamma}\right)^{3K\eta n} \over |\bar S(0)|} |\bar S|\\
&\le
O(n) \gamma^{2\lambda n}  \left({2 \over \gamma}\right)^{3K\eta n} 2^{n+K\eta n} \\
&\le \gamma^{2\lambda n-3K\eta n} 4^n \\
&\le (1/2)^n.
\end{align*}
The first inequality follows from Theorem~\ref{theorem:probability-bounds} and the fact that the total number of clauses and thus terms in 
the sum $r\ge \lambda n$ is $O(n)$.
The second inequality follows  from  (\ref{eq:size-bar-S}) and $|\bar S(0)|\ge 1$.
The  third inequality follows from (\ref{eq:parameters3}) ,and the last inequality follows from (\ref{eq:parameters2}).
Thus, the overall probability measure associated with strings
$z\in \bar S(r), r>\lambda n$ is exponentially small (at most $(1/2)^n$).

Consider now strings $z\in \bigcup_{r\le \lambda n} \bar S(r)$. The projection $x_S(z)$ to the assignment of variables in $S$ 
corresponds to the partial assignment $x: S\to \{0,1\}$ with the 
 property that it violates at most $\lambda n$ clauses out of $|C(S)|$ clauses consisting of variables in $S$.  Thus there exists an extension $\tilde x$ of
$x$ onto an assignment of variables in $[n]\setminus S$ such that the extension $\tilde x$
 belongs to one of the clusters $CL_\ell$. For any
such cluster $CL_\ell$, let $BCL_\ell$ be the set of strings $z \in \bigcup_{r\le \lambda n} \bar S(r)$ with this property. 
Namely, for every $z\in BCL_\ell$, there is an extension $\tilde x$ of $x_S(z)$ such that $\tilde x\in CL_\ell$. We now obtain an upper bound
on the cardinality of $BCL_\ell$. For any fixed assignment $x:S\to\{0,1\}$, there exists at most $2^{K\eta n}$ strings
$z$ such that $x_S(z)$ coincided with this assignment. These strings  correspond to the at most $K\eta n$ qubits in $\bigcup_{i\notin S}D(i)$.
Thus:
\begin{align*}
|BCL_\ell |\le 2^{K\eta n} |CL_\ell |\le 2^{K\eta n} 2^{c_1 n}.
\end{align*}
Applying Theorem~\ref{theorem:probability-bounds} we obtain:
\begin{align*}
\sum_{z\in BCL_\ell} | \langle \psi | z\rangle |^2
&\le 
2^{K\eta n} |CL_\ell | {\left({2 \over \gamma}\right)^{3K\eta n} \over |\bar S(0)|} \\
&\le
\left({2 \over \gamma}\right)^{4K\eta n}{ |CL_\ell |  \over |\bar S(0)|}  \\
&\le 
\left({2 \over \gamma}\right)^{4K\eta n}2^{-(c_2-c_1)n} \\
&\le 
2^{-{c_2-c_1\over 2}n},
\end{align*}
where the second-to-last inequality uses (\ref{eq:barSzeroLower}) and the last inequality uses (\ref{eq:parameters1}).

To translate the above into a circuit lower bound we apply Lemma \ref{lem:circuit_lower_bound}. Let $p(z) = \lvert\braket{\psi | z }\rvert^2$ be the distribution of bit strings resulting from measuring $\ket{\psi}$ in the computational basis. From the above, we have that the probability of measuring a bit string violating at most $\lambda n$ clauses is at least
\begin{equation}
    p\left( \bigcup_{r\le \lambda n} \bar S(r) \right)   \geq 1 - (1/2)^n, 
\end{equation}
which is greater than $0.95$ for large enough $n$. Each of the bit strings in $\bigcup_{r\le \lambda n} \bar S(r)$ is an element of a cluster $BCL_\ell$. Furthermore, we have that $p(BCL_\ell) \leq 2^{-{c_2-c_1\over 2}n}$. Therefore, for large enough $n$, we partition \change{the union of all clusters $\cup_{\ell} BCL_\ell$ into disjoint sets $S_1$ and $S_2$} such that each has size $p(S_1), p(S_2) \geq 0.45= \mu$. From the $(\nu_1,\nu_2)$-OGP, elements of $S_1$ and $S_2$ must be at least distance $d(S_1,S_2) \geq \nu_2 n$ apart as $S_1$ and $S_2$ contain disjoint clusters. Applying Lemma \ref{lem:circuit_lower_bound} we finally have that:
\begin{equation}
        T \geq \frac{1}{3}\log \left( \frac{\nu_2^2 n}{400 \log(1/0.45) } \right) = \Omega(\log(n)).
\end{equation}

\end{proof}

\section{Local Hamiltonians with bounded degree}\label{sec:bounded_degree}
One arguable limitation of our result is that in our model the degree (number of local Hamiltonians per qubit) remains bounded only on average,
while at maximum can be as large as the largest degree in our random K-SAT model ($O(\log n/\log\log n)$). We now show 
that this is fixable by switching to a different model; namely, the $p$-spin Ising model on a sparse random regular graph. We will not go into full proof 
of the result and just lay out key ideas for brevity. 

We first describe the model on a general $p$-uniform hypergraph. First recall that a $p$-uniform hypergraph
on the set of nodes $[n]$ is a collection $e_1,\ldots,e_m$ of cardinality $p$ subsets of $[n]$ called hyperedges. The case where $p=2$ corresponds
to an undirected graph in the usual sense. Suppose real values $J_e$---which we call couplings---are associated with all hyperedges $e_1,\ldots,e_m$.
We consider a (classical) Hamiltonian $H$ which associates an energy $H(\sigma)=\sum_e J_e\prod_{1\le i\le p}\sigma(e(i))$ 
with each spin assignment $\sigma:[n]\to \{\pm 1\}$. Here, $e(i)\in [n]$ denotes the $i$-node in hyperedge $e$ (the ordering is irrelevant since a product is taken). 

We now consider this model on a graph $\G_{n,d,p}$ which is a $d$-regular $p$-uniform hypergraph generated uniformly at random
from all such graphs. We call a hypergraph $d$-regular if every node is included in precisely $d$ hyperedges; note that the identity $nd=mp$ should hold. The existence of such graphs is a classical result~\cite{janson2011random}. We also assume that the couplings
 $J_e$ are generated i.i.d.\ uniformly from $\left\{\pm 1\right\}$ (namely, they are Rademacher distributed), 
 independently from everything else. We now lay out some facts about this model and either
 provide references for them or lay out road maps for proving them using methods which are now fairly standard.
 
 \begin{enumerate}
 \item[(a)] The ground state energy $\min_\sigma H(\sigma)$ satisfies 
 $\min_\sigma H(\sigma)=-\eta_{\rm Parisi}\sqrt{d}(1+o_d(1))n$ w.h.p.\ as $n\to\infty$.
 Here, $o_d(1)$ hides a factor converging to zero as $d$ increases. This fact was proven in~\cite{dembo2017extremal} for $p=2$ both for random regular
 graphs (discussed above) and also \ER graphs with average degree $d$. The  case of general $p$ was proven  
 in~\cite{chen2019suboptimality} for \ER hypergraphs, 
 and is anticipated to hold for regular graphs as well using the same ideas, though this is not documented in the literature.
 The interpolation is done 
 from the mean field model, namely, the model where all ${n\choose p}$ couplings are present (a complete hypergraph), also
 with Rademacher-distributed couplings. This interpolation was proven recently also for quantum spin glass models and variants such as the so-called
 SYK model~\cite{anschuetz2023product}.
 The validity of this interpolation can be proven using the Lindeberg's method~\cite{sen2018optimization} which also implies 
 the universality property:  the Rademacher distribution can be replaced with any zero mean, variance 1 sufficiently well-behaved distribution.

\item[(b)] The mean field models exhibit the overlap gap property when $p$ is even and satisfies $p\ge 4$. 
This was proven in~\cite{chen2019suboptimality}, and a simpler
proof is found in~\cite{gamarnik2023shattering} 
for the case of large (even or odd) $p$. Furthermore, the overlap gap property ``survives'' the  interpolation process,
and, as a result, holds in \ER and random regular hypergraphs as well for $d$ sufficiently large when $p\ge 4$ is even. 
For the \ER case this was proven in~\cite{chen2019suboptimality} and
is expected to be true for random regular hypergraphs as well using the same technique, though is not documented in the literature.

\end{enumerate}

Combining these facts we obtain a model---the $p$-spin Ising model on random regular hypergraphs---which exhibits an overlap gap property.
Next, we can construct a quantum Hamiltonian (with zero ground state energy up to a shift by the identity) from this classical Hamiltonian using the same method as for the random K-SAT Hamiltonian.
Since $d$ does not depend on the number of spins $n$, our new quantum Hamiltonian is not only $p$-local, but also has bounded (by $dp$)
degree as per our construction. This procedure thus produces a local bounded degree quantum Hamiltonian which exhibits a combinatorial NLTS property.

\section*{Acknowledgements}

The authors thank Aram Harrow and Brice Huang for valuable discussions. E.R.A.\ acknowledges funding by STAQ under award NSF Phy-1818914. D.G.\ acknowledges funding from NSF Grant DMS-2015517.

\bibliographystyle{alphaurl}
\bibliography{main}

\end{document}